\documentclass[11pt,letterpaper,reqno]{amsart}
\usepackage[T1]{fontenc}
\usepackage[utf8]{inputenc}
\usepackage{lmodern}
\usepackage[margin=1in]{geometry}
\usepackage{amsmath,amssymb,amsthm,mathtools}
\usepackage{microtype}
\usepackage{xurl}
\usepackage[colorlinks=true,linkcolor=blue,citecolor=blue,urlcolor=blue]{hyperref}
\hypersetup{pdftitle={Hypercontractivity, entropy contraction, and approximate tensorization},pdfauthor={}}
\numberwithin{equation}{section}
\newtheorem{theorem}{Theorem}[section]
\newtheorem{lemma}[theorem]{Lemma}
\newtheorem{introtheorem}{Theorem}

\newtheorem{proposition}[theorem]{Proposition}
\newtheorem{corollary}[theorem]{Corollary}
\theoremstyle{definition}
\newtheorem{definition}[theorem]{Definition}
\newtheorem{example}[theorem]{Example}
\newtheoremstyle{remarkbold}{\topsep}{\topsep}{\normalfont}{}{\bfseries}{.}{.5em}{}
\theoremstyle{remarkbold}
\newtheorem{remark}[theorem]{Remark}

\newcommand{\cM}{\mathcal M}
\newcommand{\cN}{\mathcal N}
\newcommand{\cL}{\mathcal L}
\DeclareMathOperator{\Tr}{Tr}
\DeclareMathOperator{\id}{id}
\DeclareMathOperator{\supp}{supp}
\DeclareMathOperator{\Ent}{Ent}

\title{Quantum Entropy Contraction and Factorization from Hypercontractivity}
\author{LI GAO AND LIJUN WANG}
\date{}

\begin{document}
\maketitle
\begin{abstract}We prove that hypercontractivity implies entropy contraction for a single quantum channel without assuming detailed balance. For primitive quantum Markov semigroups that are KMS-symmetric with respect to a faithful invariant state \(\sigma\), we obtain the modified logarithmic Sobolev bound
$\alpha_1\geq \frac{\lambda}{(2+\log\|\sigma^{-1}\|_\infty)}$,
where $\lambda$ is the spectral gap. This removes the assumption of \(L_p\)-regularity previously needed to obtain this bound by comparison with the logarithmic Sobolev constant. As an application, we show that hypercontractivity of the average of two conditional expectations implies approximate tensorization of relative entropy.
\end{abstract}

\section{Introduction}
Quantitative convergence to equilibrium can be studied through either norm estimates
or entropy inequalities. Hypercontractivity describes the improvement
of weighted $L_p$ norms under a Markov evolution, while entropy contraction
measures the decrease of relative entropy. In this work, we study how hypercontractive estimates imply entropy contraction in finite-dimensional quantum systems.
We note that the passage from hypercontractivity (HC) to entropy contraction (EC) already
makes sense for a single channel. This is inspired by the work of Salez~\cite{r18}
in the classical setting, where neither reversibility nor
regularity is assumed. In the quantum setting, the corresponding statement
follows from Hirche--Rouz\'e--Fran\c{c}a~\cite[Proposition~5.4]{r11} through
KMS duality. We combine it
with a centering criterion that removes the multiplicative defect in a hyperbounded
(HB) estimate using strict $L_2$-contraction on centered elements.

Let $\sigma$ be a faithful state, and let
$E_\sigma(X)=\Tr(\sigma X)I$ be the replacer map. The symmetric weighted norms are given by
$\|X\|_{p,\sigma}=\|\sigma^{1/(2p)}X\sigma^{1/(2p)}\|_p$, and $D(\rho||\sigma)=\Tr(\rho\log\rho-\rho\log \sigma )$ denotes the quantum relative entropy. We write $\Phi_*$ for the trace dual of a Heisenberg-picture channel $\Phi$.

\begin{introtheorem}[Single-channel contraction]\label{thm:intro-channel}
	Let $\Phi$ be a unital completely positive map preserving $\sigma$.
	\begin{enumerate}
		\item[(i)] If $1<p<q<\infty$ and
		$\|\Phi\|_{p\to q,\sigma}\leq1$, then, for every state $\rho$,
		\[
		D(\Phi_*\rho\Vert\sigma)\leq\frac pqD(\rho\Vert\sigma).
		\]
		\item[(ii)] 
		If, for some $2<q<\infty$, $M_q\geq1$, and $0<r<1$,
		\[
		\|\Phi\|_{2\to q,\sigma}\leq M_q,
		\qquad \|\Phi-E_\sigma\|_{2\to2,\sigma}\leq r,
		\]
		then $\|\Phi\|_{2\to p,\sigma}\leq1$ and, for every state $\rho$,
        \[D(\Phi_*\rho\Vert\sigma)\leq\frac{2}{p} D(\rho\Vert\sigma) , \]
        where
		\[
		p=2+\frac{2(q-2)\log(1/r)}{q-2+q\log(M_q/r)}>2 .
		\]
	\end{enumerate}
\end{introtheorem}
Part~(i) is established in Proposition~\ref{cor:2-4} by applying \cite[Proposition~5.4]{r11} and identifying the relevant
KMS adjoint with the trace adjoint of the Petz map. Part~(ii), proved in
Theorem~\ref{prop:2-9}, is a weighted version of the centering method of
Wang~\cite[Theorem~1.6]{Wang2016}. In finite dimensions, a finite HB
bound is automatic; its quantitative value determines the HC exponent supplied
by the centering method. These statements require neither symmetry nor a semigroup
structure and provide the direct norm-to-entropy passage used below.

For a primitive quantum Markov semigroup (QMS) $P_t=e^{-t\cL}$, we write $\alpha_1$, $\alpha_2$, and $\alpha_H$ for the constants associated with the modified logarithmic Sobolev inequality (MLSI), the $L_2$-logarithmic Sobolev inequality (LSI), and
HC, respectively; see
Section~\ref{subsec:3-1} for detailed definitions. In the continuous-time setting, we obtain the following entropy contraction bounds.

\begin{introtheorem}[MLSI from hyperboundedness]\label{thm:intro-semigroup}
	Let $(P_t)$ be a finite-dimensional primitive KMS-symmetric QMS with a faithful
	invariant state $\sigma$ and spectral gap $\lambda>0$. If, for some $2<q\le \infty,\ t_q>0,\ M_q\geq1$,
	\[
	\|P_{t_q}\|_{2\to q,\sigma}\leq M_q,
	\]
	then
	\[
	\alpha_1\geq\beta,\qquad \alpha_2\geq2\beta,\qquad
	\alpha_H\geq\beta,\qquad
	\beta=\frac{(1-2/q)\lambda}
	{2[\lambda t_q+\log M_q+1-2/q]}.
	\]
Moreover, for a primitive QMS with positive KMS coercivity constant \(\lambda_{\mathrm{KMS}}>0\), without assuming symmetry,
\begin{equation}\label{eq:intro-mlsi}
\alpha_1\geq\frac{\alpha_H}{2}\geq
\frac{\lambda_{\mathrm{KMS}}}{2+\log\|\sigma^{-1}\|_\infty}.
\end{equation}
\end{introtheorem}
We emphasize that no $L_p$-regularity is assumed in Theorem~\ref{thm:3-endpoint}.
This estimate is proved by constructing a local HC
curve for small $t$ from the endpoint estimate and the spectral gap. Its slope at $t=0$ gives the
MLSI and LSI bounds; extending the curve to all times gives the 
HC bound. The interpolation method, the existence of an all-time HC curve, and the implication from HC to LSI without
$L_p$-regularity are already present in~\cite{r20}. Here the centering criterion in Lemma~\ref{lem:2-8}
provides explicit constants, and the HC-to-EC implication gives MLSI directly.
The quantitative comparison with~\cite{r20} is given in
Section~\ref{subsec:3-bounds}.

For KMS-symmetric semigroups, bounds similar to \eqref{eq:intro-mlsi} were obtained for the LSI constant $\alpha_2$, but the comparison with the MLSI constant has required $L_p$-regularity. Here, we remove this technical assumption and obtain the bound for all KMS-symmetric semigroups.
Related qualitative MLSI results for primitive KMS-symmetric semigroups also recently
appeared in~\cite{LWW2026,GG2026}.

Let $E_i:\cM\to\cN_i$, $i=1,2$, be conditional expectations preserving a common
faithful state $\sigma$, and let $E_{\cN}$ be the corresponding expectation onto
$\cN=\cN_1\cap\cN_2$. Write 
\[
\Phi_{\mathrm{av}}=\frac{E_1+E_2}{2}.
\]
Approximate tensorization (AT) means that there is a constant $c>0$ such that
$$D(\rho||E_{\cN*}\rho )\le c\left(D(\rho\Vert E_{1*}\rho)+D(\rho\Vert E_{2*}\rho)\right)$$
for all states $\rho$. When $\cN\ne\cM$, AT with constant $c>0$ implies EC for
$\Phi_{\mathrm{av},*}$ with coefficient $1-1/(2c)$, by the chain rule and convexity~\cite{r4,r7}. Our third theorem proves a converse whose AT
constant depends only on the averaged EC coefficient.

\begin{introtheorem}[AT from averaged EC]\label{thm:intro-at}
Let $E_i:\cM\to\cN_i$, $i=1,2$, be conditional expectations preserving a common
faithful state $\sigma$, and let $E_{\cN}$ be the corresponding expectation onto
$\cN=\cN_1\cap\cN_2$.	If, for some $0\leq\eta<1$ and every state $\rho$,
	\[
	D(\frac{E_{1*}\rho+E_{2*}\rho}{2}\|E_{\cN*}\rho )\leq\eta D(\rho||E_{\cN*}\rho),
	\]
	then
	\begin{equation}\label{eq:intro-at}
		D(\rho|| E_{\cN*}\rho )\leq\frac{3}{2(1-\eta)}
		\left( D(\rho\Vert E_{1*}\rho) +D(\rho\Vert E_{2*}\rho) \right).
	\end{equation}
\end{introtheorem}
This theorem makes averaged EC and
AT quantitatively equivalent up to universal constants. The proof uses the quantum Jensen--Shannon divergence and the triangle inequality for the square root of
this divergence~\cite{r22}.  Approximate tensorization for noncommuting conditional expectations and its
complete versions were studied in~\cite{r4,r10}. Gao and
Rouz\'e~\cite[Corollary~5.4 and Remark~5.5]{r10} also obtain index-independent
bounds from mixing in completely positive order, a condition of ultracontractive $L_1\to L_\infty$ type; an index factor enters when
one converts $L_2$ estimates into that mixing condition. Our result starts
with an entropy contraction estimate and converts it to AT without an
additional index factor.  As an application of Theorem~\ref{thm:intro-channel},
when $\cN=\mathbb CI$, the hypercontractivity of the average channel $\rho\mapsto \frac{1}{2}(E_{1*}\rho+E_{2*}\rho)$ 
supplies the EC
hypothesis needed to obtain AT.

Section~\ref{sec:2} develops the single-channel estimates.
Section~\ref{sec:3} treats quantum Markov semigroups and discusses
entropy decay from time zero as well as delayed estimates.
Section~\ref{sec:4} proves the implication from averaged EC to AT.

\section{Entropy contraction via hypercontractivity}\label{sec:2}
We derive entropy contraction from hypercontractivity using a dual form of \cite[Proposition~5.4]{r11}. Hypercontractivity can be obtained from a hyperbounded endpoint and strict $L_2$ contraction on centered observables, using a weighted version of the method in \cite[Theorem~1.6]{Wang2016}.
The centered criterion will be used directly for semigroups in Section~\ref{sec:3}.

\subsection{Weighted norms and state-preserving channels}\label{subsec:2-1}
Throughout the paper, let $B(\mathcal{H})$ denote the algebra of linear operators on a finite-dimensional Hilbert space $\mathcal{H}$ and let $\Tr$ denote the standard matrix trace. We identify states with density matrices and fix a faithful state $\sigma>0$, $\Tr\sigma=1$. All logarithms are natural.

For $1\le p<\infty$, recall the symmetric weighted noncommutative $L_p$ norms of \cite[equations~(10)--(12)]{r13}, defined by
\begin{equation}\label{eq:2-1}
 \|X\|_{p,\sigma}=\bigl\|\sigma^{1/(2p)}X\sigma^{1/(2p)}\bigr\|_p,
 \qquad \|X\|_{\infty,\sigma}=\|X\|_\infty,
\end{equation}
where the $p$-norm on the right is the Schatten $p$-norm. We write $L_p(\sigma)$ for this normed space, and $\|T\|_{p\to q,\sigma}$ for the induced operator norm; the subscript $\sigma$ is omitted whenever there is no ambiguity. For $p=2$, $L_2(\sigma)$ is a Hilbert space equipped with the Kubo–Martin–Schwinger (KMS) inner product
\begin{equation}\label{eq:2-2}
 \langle X,Y\rangle_\sigma=\Tr(X^*\sigma^{1/2}Y\sigma^{1/2}),
 \qquad \Gamma_\sigma(X)=\sigma^{1/2}X\sigma^{1/2},
\end{equation}
where $\Gamma_\sigma$ is the weighting map.
For real $s$ we write $\Gamma_\sigma^s(X)=\sigma^{s/2}X\sigma^{s/2}$. For conjugate exponents $1/p+1/p'=1$, weighted H\"older duality gives
\begin{equation}\label{eq:2-3}
 |\langle X,Y\rangle_\sigma|\le\|X\|_{p,\sigma}\|Y\|_{p',\sigma},
 \qquad \|X\|_{p,\sigma}=\sup_{\|Y\|_{p',\sigma}\le1}|\langle Y,X\rangle_\sigma|.
\end{equation}
The norms are monotone: $\|X\|_{p,\sigma}\le\|X\|_{q,\sigma}$ for $p\le q$.  The spaces $L_p(\sigma)$ form a complex interpolation scale \cite[Theorem~4]{r5},
\begin{equation}\label{eq:2-4}
 \|T\|_{p_\theta\to q_\theta}\le
 \|T\|_{p_0\to q_0}^{1-\theta}\|T\|_{p_1\to q_1}^{\theta},
 \quad \frac1{p_\theta}=\frac{1-\theta}{p_0}+\frac\theta{p_1},
 \quad \frac1{q_\theta}=\frac{1-\theta}{q_0}+\frac\theta{q_1}.
\end{equation}
We will also use the analytic-family version of Riesz--Thorin interpolation in Section~\ref{sec:3}.

A quantum channel in the Heisenberg picture is a unital completely positive map $\Phi:B(\mathcal{H})\to B(\mathcal{H})$. Its trace dual $\Phi_*$ in the Schr\"odinger picture is a completely positive trace-preserving (CPTP) map given by
\[
 \Tr(\Phi_*\rho\,X)=\Tr(\rho\,\Phi X).
\]
Throughout the paper, we assume that $\Phi$ is $\sigma$-preserving, i.e. $\sigma=\Phi_*\sigma$ is an invariant state. Every $\sigma$-preserving channel is contractive
\[ \|\Phi(X)\|_{p, \sigma}\le \|X\|_{p, \sigma}\]
on $L_p(\sigma)$ for $1\le p\le\infty$. %Indeed, unital positivity gives the $p=\infty$ endpoint, and duality with the $\infty$-contractivity of $\Phi^\#$ gives the $p=1$ endpoint. Complex interpolation gives the remaining exponents.
The KMS adjoint with respect to $L_2(\sigma)$ is
\begin{equation}\label{eq:2-5}
 \Phi^\#=\Gamma_\sigma^{-1}\Phi_*\Gamma_\sigma.
\end{equation}
This is exactly the Petz recovery map of $\Phi_*$ with respect to the invariant state $\sigma$ in the Heisenberg picture. In particular,  $\Phi^\#$ is unital and completely positive and satisfies $\sigma\circ \Phi^\#=\sigma $. We call $\Phi$ KMS-symmetric if $\Phi^\#=\Phi$.

\iffalse

For completely positive maps, the positive-input reduction \cite[Theorem~1 and the proof of Lemma~2]{r23} states that induced Schatten $p\to q$ norms can be computed on positive inputs. For weighted norms this applies to the completely positive map $\Gamma_\sigma^{1/q}T\Gamma_\sigma^{-1/p}$. Consequently, a positive-observable HC estimate for a completely positive map gives the full operator-norm estimate used here. For direct-sum algebras, extend the sandwiched map to a full matrix algebra by composing with the block pinching. Since that pinching is a Schatten contraction and commutes with the sandwich weights, the induced cross-norm is unchanged. The same pinching extends a $\sigma$-preserving channel to the ambient full matrix algebra without changing its weighted norm bound.
\fi 
\subsection{From norm estimates to entropy contraction}\label{subsec:2-2}

For two states $\rho,\omega$, their Umegaki relative entropy \cite{r21} is
\[
 D(\rho\|\omega)=\Tr\rho(\log\rho-\log\omega)
\]
if $\supp\rho\subseteq\supp\omega$, and is $+\infty$ otherwise. 
We recall three functional inequalities for state-preserving quantum channels.

\begin{definition}[Hypercontractivity, hyperboundedness and entropy contraction]\label{def:2-2}
Let $\Phi: B(\mathcal{H})\to B(\mathcal{H})$ be a unital completely positive map and assume $\Phi$ preserves a fixed faithful state $\sigma=\Phi_*\sigma$. For $1<p<q<\infty$, we say that
\begin{enumerate}
\item[i)] $\Phi$ is \emph{hypercontractive} (HC) at $(p,q)$ if $$\|\Phi\|_{p\to q,\sigma}\le1. $$
\item[ii)] $\Phi$ is \emph{hyperbounded} (HB) at $(p,q)$ if for some $1\le M<\infty$,
$$\|\Phi\|_{p\to q,\sigma}\le M.$$
\item[iii)] $\Phi$ satisfies \emph{(relative) entropy contraction} with respect to $\sigma$ if for some $0<\eta<1$,
\begin{align}\label{eq:EC}
D(\Phi_*\rho\|\sigma)\le \eta D(\rho\|\sigma),\qquad \forall  \  \text{ state }\rho. \end{align}
\end{enumerate}
\end{definition}
\iffalse
 Following \cite[equation~(1.3)]{r11}, the fixed-state relative-entropy contraction coefficient, also called the strong data-processing coefficient, is
\begin{equation}\label{eq:2-12}
 \eta_{\mathrm{Re}}(\Phi_*;\sigma)=\sup_{\rho\ne\sigma}
 \frac{D(\Phi_*\rho\|\sigma)}{D(\rho\|\sigma)}.
\end{equation}
Thus EC with coefficient $\eta<1$ asserts that this supremum is at most $\eta$.
\fi
In finite dimensions, a hyperbounded estimate always holds and the information lies in the value of $M$ and the dependence on the parameters $(p,q)$. Hypercontractivity is the case $M=1$, which can hold only if $\Phi$ is primitive, i.e. $\Phi^k$ converges to $E_\sigma$. In particular, $\sigma$ is the unique invariant state and the fixed-point space $\mathcal{N}:=\{ X \ |\  \Phi(X)=X\}=\mathbb{C}1$ is trivial. By the data processing inequality, the relative entropy is always monotone non-increasing 
\[ D(\Phi_*\rho\|\sigma)\le D(\rho\|\sigma)\ .\]
The entropy contraction \eqref{eq:EC} asserts a uniform contraction factor strictly smaller than one and is also called a strong data processing inequality (SDPI).

 Since $\Phi_*\sigma=\sigma$, the Petz map of $\Phi_*$ at $\sigma$ satisfies
\begin{equation}\label{eq:2-14}
 R_{\sigma,\Phi_*}=\Gamma_\sigma\Phi\Gamma_\sigma^{-1},
 \qquad R_{\sigma,\Phi_*}^*=\Gamma_\sigma^{-1}\Phi_*\Gamma_\sigma=\Phi^\#.
\end{equation}
Here $*$ denotes the trace adjoint. Weighted H\"older duality \eqref{eq:2-3} gives
\begin{equation}\label{eq:2-15}
 \|\Phi\|_{p\to q,\sigma}=\|\Phi^\#\|_{q'\to p',\sigma},
 \qquad \|\Phi^\#\|_{1\to1,\sigma}=1.
\end{equation}
Applying \cite[Proposition~5.4]{r11} using the KMS duality above, we obtain a quantum analogue of Salez's result \cite[Theorem~1]{r18} for classical Markov maps. No symmetry (reversibility) is required, as in the classical case \cite{r18}.
\iffalse
The next proposition is \cite[Proposition~5.4]{r11} in the weighted Heisenberg picture. Its KMS-symmetric conclusion follows by applying the same result to the dual exponents. We include the interpolation proof to fix the coefficients and conventions.
\fi

\begin{proposition}[Entropy contraction from hypercontractivity]\label{cor:2-4}
Let $1<p\le q<\infty$. Suppose $\Phi$ is a $\sigma$-preserving channel satisfying
\begin{equation}\label{eq:2-16}
 \|\Phi\|_{p\to q,\sigma}\le 1.
\end{equation}
Then, for every state $\rho$,
\begin{equation}\label{eq:2-17}
 D(\Phi_*\rho\|\sigma)\le\frac pqD(\rho\|\sigma).
\end{equation}
\iffalse
In particular, if $M=1$ and $p<q$, then
\begin{equation}\label{eq:2-18}
 \eta_{\mathrm{Re}}(\Phi_*;\sigma)\le\frac pq<1.
\end{equation}
If, in addition, $\Phi^\#=\Phi$, then
\begin{equation}\label{eq:2-19}
 \begin{split}
 D(\Phi_*\rho\|\sigma)\le\min\Bigl\{&\frac pqD(\rho\|\sigma)+p\log M,\\
 &\frac{q'}{p'}D(\rho\|\sigma)+q'\log M\Bigr\}.
 \end{split}
\end{equation}
Thus in the KMS-symmetric hypercontractive case,
\begin{equation}\label{eq:2-20}
 \eta_{\mathrm{Re}}(\Phi_*;\sigma)\le\min\left\{\frac pq,\frac{q'}{p'}\right\}.
\end{equation}
\fi
\end{proposition}
\begin{proof}
For $0\le t\le1$, set
\[
 a_t=(1-t/q)^{-1},\qquad b_t=(1-t/p)^{-1}.
\]
By \eqref{eq:2-15} and \eqref{eq:2-16}, complex interpolation yields
$$\|\Phi^{\#}\|_{a_t\to b_t,\sigma}\le 1. $$
For a faithful state $\rho$, take $X=\Gamma_\sigma^{-1}\rho$.
Then $\Phi^\# X=\Gamma_\sigma^{-1}\Phi_*\rho$, and both have unit weighted $L_1$ norm.
The norm--entropy derivative \cite[Lemma~5.1]{r11} is
\[
 \left.\frac{d}{ds}\log\|\Gamma_\sigma^{-1}\rho\|_{s,\sigma}\right|_{s=1}
 =D(\rho\Vert\sigma).
\]
Taking logarithms and differentiating the interpolated estimate at $t=0$ gives
\[
 \frac1pD(\Phi_*\rho\Vert\sigma)
 \le\frac1qD(\rho\Vert\sigma),
\]
which is \eqref{eq:2-17}. The result for a general state $\rho$ follows by approximating it with the faithful states $(1-\varepsilon)\rho+\varepsilon\sigma$.
\end{proof}
Hypercontractivity also implies primitivity. Recall that a channel is primitive if its iterates $\Phi^k$ converge to the scalar expectation $E_\sigma(X)=\Tr(\sigma X)I$. If $\|\Phi\|_{p\to q,\sigma}\le1$ with $p<q$, iteration gives
\begin{equation}\label{eq:2-21}
 D(\Phi_*^k\rho\|\sigma)\le\left(\frac pq\right)^kD(\rho\|\sigma)\to 0, 
 \qquad \text{ as } k\to \infty.
\end{equation}
By quantum Pinsker's inequality \cite[Section~I]{r13}, this yields trace norm convergence
$$\|\Phi_*^k\rho-\sigma\|_1^2\le2\left(\frac pq\right)^k D(\rho\|\sigma).$$ 

When the fixed-point algebra $\cN$ is nontrivial, relative hyperboundedness was introduced in \cite{r3} using amalgamated norms relative to $\cN$. If  $E$ is the fixed-point conditional expectation and $\Phi E=E\Phi=E$, the corresponding entropy contraction is formulated using $D(\rho\|E_*\rho)$, the decoherence-free relative entropy introduced in \cite[Definition~3.3]{r2}, whose reference state $E_*\rho$ depends on the input $\rho$. However, for the amalgamated-norm formulation of \cite{r3}, a nontrivially decohering QMS on $B(\mathcal H)$ cannot satisfy a defect-free HC curve with positive initial slope when $\mathbb CI\subsetneq\cN\subsetneq B(\mathcal H)$; see \cite[Theorems~2.4(i) and~5.1]{r3}.

\subsection{From hyperboundedness to hypercontractivity}\label{subsec:2-3}
It follows by an argument similar to that of Proposition \ref{cor:2-4} that hyperboundedness implies entropy contraction with a defect
\[ \|\Phi\|_{p\to q,\sigma}\le M \qquad \Longrightarrow \qquad  D(\Phi_*\rho\|\sigma)\le\frac pqD(\rho\|\sigma)+p\log M \ .  \]
We now use strict $L_2$ coercivity to remove the defect term. The key ingredient is an $L_p$-norm analogue of the Rothaus lemma, based on a centering argument. We state the centering argument separately, as it will be reused in Section~\ref{sec:3}.

Define the expectation map
\begin{equation}\label{eq:2-10}
 E_\sigma(X)=\Tr(\sigma X)I.
\end{equation}
This gives the KMS orthogonal decomposition $L_2(\sigma)=\mathbb CI\oplus \ker E_\sigma$, with 
$$\|aI+X_0\|_{2,\sigma}^2=|a|^2+\|X_0\|_{2,\sigma}^2$$ for $X_0\in \ker E_\sigma$.
Every $\sigma$-preserving channel satisfies $\Phi E_\sigma=E_\sigma\Phi=E_\sigma$, and therefore
\begin{equation}\label{eq:2-11}
 \|\Phi-E_\sigma\|_{2\to2,\sigma}=\|\Phi|_{\ker E_\sigma}\|_{2\to2,\sigma}.
\end{equation}

\begin{lemma}\label{lem:2-8}
Let $\Phi$ be a $\sigma$-preserving channel and $2\le p<\infty$. If
\begin{equation}\label{eq:2-23}
 \|\Phi-E_\sigma\|_{2\to p,\sigma}\le(p-1)^{-1/2},
\end{equation}
then $\|\Phi\|_{2\to p,\sigma}\le1$.
\end{lemma}
\begin{proof}
Lemma~2.9 of \cite{r29}, in its reversed form for $p>2$, gives
\begin{equation}\label{eq:2-22}
 \|X\|_{p,\sigma}^2\le|\Tr(\sigma X)|^2
 +(p-1)\|X-E_\sigma X\|_{p,\sigma}^2.
\end{equation}
Writing $X=aI+X_0$, with $X_0\in \ker E_\sigma$, we obtain
\begin{align*}
 &\|\Phi X\|_{p,\sigma}^2
 \le|a|^2+(p-1)\|\Phi X_0\|_{p,\sigma}^2
 \le|a|^2+\|X_0\|_{2,\sigma}^2
 =\|X\|_{2,\sigma}^2 . \qedhere
\end{align*}
\end{proof}

The next theorem is a weighted version of \cite[Theorem 1.6]{Wang2016}. It upgrades a hyperbounded estimate to hypercontractivity via centered $L_2$ contraction and hence gives entropy contraction.

\begin{theorem}\label{prop:2-9}
Let $\Phi$ be a $\sigma$-preserving channel. Suppose that \begin{enumerate}
\item[i)] for some $2<q<\infty$ and $M_q\ge1$,
\begin{equation}\label{eq:2-24}
 \|\Phi\|_{2\to q,\sigma}\le M_q,
\end{equation}
\item[ii)] for some $0<r<1$,
 \begin{align} \|\Phi-E_\sigma\|_{2\to2,\sigma}\le r, \label{eq:gap}\end{align}
\end{enumerate}
Then 
\begin{align}\label{eq:2-26}
 \|\Phi\|_{2\to p,\sigma}\le1,
\end{align} 
and for every state $\rho$,
\begin{align}\label{eq:2-26x}
\qquad D(\Phi_*\rho\|\sigma)\le\frac2pD(\rho\|\sigma)
\end{align}
where
\begin{equation}\label{eq:2-27}
 p=2+\frac{2(q-2)\log(1/r)}{q-2+q\log(M_q/r)}.
\end{equation}
\end{theorem}
\begin{proof}
For $0<r<1$ and $M_q\ge1$, the definition of $p$ gives
$0<p-2<2(q-2)/q<q-2$. For $2<p\le q$, put $\theta_p=(1/2-1/p)/(1/2-1/q)$.
For $X$ with $E_\sigma(X)=0$, interpolation of the output norms gives
\[
 \|\Phi X\|_{p,\sigma}
 \le\|\Phi X\|_{2,\sigma}^{1-\theta_p}\|\Phi X\|_{q,\sigma}^{\theta_p}
 \le r^{1-\theta_p}M_q^{\theta_p}\|X\|_{2,\sigma}.
\]
Thus Lemma~\ref{lem:2-8} gives hypercontractivity whenever
\begin{equation}\label{eq:2-25}
 (p-1)r^{2(1-\theta_p)}M_q^{2\theta_p}\le1.
\end{equation}
By the inequalities $\log(p-1)\le p-2$ and
$\theta_p\le q(p-2)/(2(q-2))$, the logarithm of the left-hand side is bounded by
\[
  \log \left((p-1)r^{2(1-\theta_p)}M_q^{2\theta_p}\right)\le   2\log r+\left(1+\frac q{q-2}\log(M_q/r)\right)(p-2)\le0,
\]
where $p$ is chosen as in \eqref{eq:2-27} to ensure the last inequality.
This proves hypercontractivity at $(2,p)$, which gives entropy contraction \eqref{eq:2-26x} by
Proposition~\ref{cor:2-4}.
\end{proof}
\iffalse
For $0<r<1$, the test \eqref{eq:2-25} gives a sharper sufficient range than the explicit bound above.
Write $M_q=e^{d(1-2/q)}$, $d\ge0$, and $\gamma=-\log r$.
There is a unique $p_*\in(2,q)$ such that
\begin{equation}\label{eq:2-28}
 \frac{2\gamma(q-p_*)}{q-2}=d(p_*-2)+\frac{p_*}{2}\log(p_*-1).
\end{equation}
The left side strictly decreases from $2\gamma$ to zero and the right side strictly increases from zero.
Multiplying the logarithm of \eqref{eq:2-25} by $p/2$ shows that precisely $2<p\le p_*$ satisfy the test; in particular, $p_0\le p_*$.
This is a sufficient range, not a characterization of the channel's optimal HC exponent.
\fi 

\begin{remark}{\rm a) For a primitive channel $\Phi$ that is KMS-symmetric with respect to $\sigma$, condition ii) \eqref{eq:gap} is always satisfied. Indeed, under KMS symmetry, $\Phi$ is a self-adjoint operator on $L_2(\sigma)$ commuting with $E_\sigma$, and hence the norm of $\Phi-E_\sigma$ equals its spectral radius, which is strictly less than 1 by primitivity.

b)
The above qualitative implication from strict centered $L_2$ contraction to a nontrivial HC exponent is known in the classical setting \cite[Theorem~3(a),~(c), p.~928]{r1}. 
In finite-dimensional noncommutative algebras with a tracial state, Wang \cite[Theorem~1.6]{Wang2016} proves the corresponding $s\to2$ characterization for $1<s<2$. Our argument above is a weighted version in the dual range using the same $p$-convexity inequality. }
\end{remark}

\section{Quantum Markov semigroups}\label{sec:3}
In this section, we discuss how a hypercontractive or hyperbounded estimate implies entropy contraction in continuous time.

\subsection{Functional inequalities}\label{subsec:3-1}

A \emph{quantum Markov semigroup} (QMS) is a continuous semigroup $(P_t)_{t\geq0}:B(\mathcal{H})\to B(\mathcal{H})$ of
unital completely positive maps satisfying $P_t\circ P_s= P_{t+s}$ for $s,t\ge 0$ and $P_0=\id$. Throughout this section, we assume $\mathcal{H}$ is
finite-dimensional, and the semigroup is primitive with a faithful invariant state $\sigma$.
We write $P_t=e^{-t\cL}$ with the generator
\[ \cL X=-\lim_{t\to0}\frac{1}{t}(P_tX-X) \ ,\]
and $K=\ker E_\sigma$. KMS symmetry is assumed only where stated,
and we omit $\sigma$ from $p$-norms when it is clear from the context.
For a faithful state $\rho$, the entropy production is
\cite[equations~(5)--(6)]{r13}
\begin{equation}\label{eq:3-1}
 \mathcal I_{\cL}(\rho)
 =\Tr[(\cL_*\rho)(\log\rho-\log\sigma)]
 =-\left.\frac{d}{dt}D(P_{t*}\rho\Vert\sigma)\right|_{t=0}.
\end{equation}
The HC constant $\alpha_H$ is the supremum of $\alpha \geq0$ such that
\begin{equation}\label{eq:3-2}
 \|P_t\|_{2\to1+e^{2\alpha t},\sigma}\leq1\qquad(t\geq0).
\end{equation}
For $X>0$, write $Y=\Gamma_\sigma^{\frac{1}{2}}(X)=\sigma^{\frac14}X\sigma^{\frac14}$, and set
\[
 \Ent_{2,\sigma}(X)
 =\frac12\{\Tr[Y^2(\log Y^2-\log\sigma)]-\Tr(Y^2)\log \Tr(Y^2)\},\qquad
 \mathcal E(X)=\operatorname{Re}\langle X,\cL X\rangle_\sigma.
\]
Define the constants $\alpha_2$ and $\alpha_1$ as the best nonnegative constants in the
$L_2$-logarithmic Sobolev inequality (LSI) and the modified logarithmic Sobolev inequality (MLSI), respectively:
\begin{align}\label{eq:3-3}
 \alpha_2\Ent_{2,\sigma}(X)&\leq \mathcal E(X), \tag{LSI}\\
 2\alpha_1D(\rho\Vert\sigma)&\leq\mathcal I_{\cL}(\rho). \tag{MLSI}
\end{align}
By Gronwall's lemma, MLSI is equivalent to exponential entropy decay from time zero \cite[equations~(5)--(6)]{r13}:
\begin{equation}\label{eq:3-10}
 D(P_{t*}\rho\Vert\sigma)\leq e^{-2\alpha_1 t}D(\rho\Vert\sigma),\quad t\geq0.
\end{equation}
We also use the norm--entropy derivative \cite[Theorem~4]{r13}:
\begin{equation}\label{eq:3-5}
 \left.\partial_r\log\|X\|_{r,\sigma}\right|_{r=2}
 =\frac{\Ent_{2,\sigma}(X)}{2\|X\|_{2,\sigma}^2},\qquad X>0,
\end{equation}
which gives $\alpha_H\leq \alpha_2$ \cite[Theorem~15(1)]{r13}. Direct Dirichlet-form comparison under KMS symmetry gives
$\alpha_1\geq\alpha_2$ under strong $L_p$-regularity \cite[Propositions~8 and~13]{r13}, a condition satisfied by GNS-symmetric semigroups but not known to follow from KMS symmetry. In summary, for a primitive KMS-symmetric semigroup
\[\alpha_H\le \alpha_2\overset{(*)}{\le}\alpha_1 \]
where (*) requires strong $L_p$-regularity.

The next lemma shows that $\alpha_H\leq \alpha_{2}$ and $\alpha_H\leq 2\alpha_{1}$ without KMS symmetry or $L_p$-regularity. For KMS-symmetric semigroups, the argument from HC to MLSI was discussed in
\cite[Section~7, EQ.~(7.2)--(7.3)]{r11}. 

\begin{lemma}[Local HC curve]\label{lem:3-local}
Let $r:[0,\varepsilon)\to[2,\infty)$ be right differentiable at zero, with
$r(0)=2$ and $r'(0)=v>0$. If $\|P_s\|_{2\to r(s)}\leq1$ for
$0\leq s<\varepsilon$, then
\[
 \alpha_1\geq\frac v4 \ ,\  \alpha_2 \geq \frac v2 \ ,\ \alpha_H\geq\frac v4\ .
\]
In particular, we have $\alpha_H\le 2\alpha_1$ and $\alpha_H\le \alpha_2$.
\end{lemma}

\begin{proof}
Proposition~\ref{cor:2-4} gives
$$D(P_{s*}\rho\Vert\sigma)\leq \frac{2}{r(s)}D(\rho\Vert\sigma).$$
Since $2/r(s)=1-vs/2+o(s)$, differentiation at zero gives $\alpha_1\geq v/4$.
For the LSI constant, the norm--entropy derivative \eqref{eq:3-5} gives
\[
0\geq\left.\frac{d}{ds}\log\|P_sX\|_{r(s),\sigma}\right|_{s=0^+}
=\frac{(v/2)\Ent_{2,\sigma}(X)-\mathcal E(X)}{\|X\|_{2,\sigma}^2}.
\]
Thus $v\Ent_{2,\sigma}(X)\leq2\mathcal E(X)$, which gives $\alpha_2\geq v/2$.
For the all-time estimate, interpolation with $\|P_s\|_{\infty\to\infty}=1$ gives
\[
 \|P_s\|_{p\to pr(s)/2}\leq1,\qquad 2\leq p<\infty.
\]
For fixed $t>0$, composition over $n$ steps of length $t/n$ yields
\[
 \|P_t\|_{2\to2(r(t/n)/2)^n}\leq1,
\]
where the target exponent converges to $2e^{vt/2}$ as $n\to \infty$. The bound
on $\alpha_H$ follows from $1+e^{vt/2}\leq2e^{vt/2}$ and norm monotonicity.
\end{proof}

\iffalse
For completeness, the same entropy argument works at any input exponent $1<p<\infty$.
If $q(0)=p$, $\dot q(0)>0$, and
\begin{equation}\label{eq:3-6}
 \|P_t\|_{p\to q(t)}\leq1\qquad(0\leq t<\varepsilon),
\end{equation}
for a curve right differentiable at zero, then
\begin{equation}\label{eq:3-7}
 \alpha_1\geq\frac{\dot q(0)}{2p}.
\end{equation}
Under KMS symmetry, the second coefficient in Proposition~\ref{cor:2-4} improves this to
\begin{equation}\label{eq:3-8}
 \alpha_1\geq\frac{\dot q(0)}{2p}\max\left\{1,\frac1{p-1}\right\}.
\end{equation}
Indeed, the two EC coefficients have expansions
\[
 \frac p{q(t)}=1-\frac{\dot q(0)}p t+o(t),\qquad
 \frac{q(t)(p-1)}{p(q(t)-1)}=1-\frac{\dot q(0)}{p(p-1)}t+o(t).
\]
For $q(t)=1+(p-1)e^{4at}$, these give, respectively,
\begin{equation}\label{eq:3-9}
 \alpha_1\geq\frac{2a(p-1)}p,\qquad
 \alpha_1\geq\frac{2a}p\max\{p-1,1\}\quad\text{under KMS symmetry}.
\end{equation}
\fi 

%These are the normalizations of \cite[Definition~11]{r13} and
%\cite[equations~(6)--(8)]{r20}; the entropy convention in \cite[Section~2.2]{r6} is
%$2\Ent_{2,\sigma}$. The HC parameter $\alpha_H$ is the parameter $\alpha$
%in a curve $2\to1+e^{2\alpha t}$. All inequalities in this section are ordinary,
%unamplified inequalities.
\subsection{Modified logarithmic Sobolev bounds}\label{subsec:3-bounds}

For a finite-dimensional primitive QMS $P_t=e^{-\cL t}$ with $\lim_{t\to \infty }P_t=E_\sigma$, we define the KMS coercivity constant
\[
 \lambda_{\text{KMS}}:=\inf_{0\ne X\in\ker E_\sigma}
 \frac{\operatorname{Re}\langle X,\cL X\rangle_\sigma}{\|X\|_{2,\sigma}^2}\geq0.
\]
Under KMS symmetry, $\cL$ is self-adjoint on $L_2(\sigma)$ and $\lambda_{KMS}$ is exactly the spectral gap
\begin{equation}\label{eq:3-4}
 \lambda_{\text{KMS}}=\inf_{0\ne X\in\ker E_\sigma}
 \frac{\langle X,\cL X\rangle_\sigma}{\|X\|_{2,\sigma}^2}=\min (\text{spec}(\cL) \setminus \{0\})=:\lambda>0,
\end{equation}
which is always positive in finite dimensions,  
and one has \cite[Theorem~16]{r13}
\begin{equation}\label{eq:3-12}
 \alpha_1\leq\lambda.
\end{equation}
For a finite-dimensional QMS without KMS symmetry, one has $\alpha_1\leq 2 \lambda_{\text{KMS}}$, with a factor of $2$ lost through symmetrization.

We derive the modified logarithmic Sobolev inequality from hyperboundedness using the coercivity constant or spectral gap on centered elements. The argument applies with or without KMS symmetry and does not require $L_p$-regularity.

\begin{theorem}[Bounds from centered coercivity]\label{thm:3-coercivity}
Let $(P_t)$ be a $\sigma$-preserving primitive QMS with a positive KMS coercivity constant $\lambda_{KMS}=\lambda>0$.
Then
\begin{equation}\label{eq:3-24}
 \alpha_1\geq\frac{\alpha_H}{2}\geq\frac{\lambda}{2+\log \|\sigma^{-1}\|_\infty}.
\end{equation}
\end{theorem}
\begin{proof}
For $X\in \ker E_\sigma$, coercivity gives
$$\frac{d}{dt}\|P_tX\|_{2,\sigma}^2\leq-2\lambda\|P_tX\|_{2,\sigma}^2,$$
which implies $\|P_t-E_\sigma\|_{2\to2}\leq e^{-\lambda t}$.
The elementary weighted norm comparison is
\begin{equation}\label{eq:3-22}
 \|X\|_{r,\sigma}\leq \|\sigma^{-1}\|_\infty^{1/2-1/r}\|X\|_{2,\sigma},\qquad 2\leq r<\infty.
\end{equation}
Indeed, put $Z=\sigma^{1/4}X\sigma^{1/4}$ and
$B=\sigma^{1/(2r)-1/4}$. The ideal property and Schatten norm monotonicity give
\begin{align*} \|X\|_{r,\sigma}=&\|\sigma^{1/(2r)-1/4}(\sigma^{1/4}X\sigma^{1/4}) \sigma^{1/(2r)-1/4}\|_r\leq\|\sigma^{1/(2r)-1/4}\|_\infty^2\|\sigma^{1/4}X\sigma^{1/4}\|_r
\\ \leq &\|\sigma^{-1}\|_\infty^{1/2-1/r}\|\sigma^{1/4}X\sigma^{1/4}\|_2=\|\sigma^{-1}\|_\infty^{1/2-1/r}\|X\|_{2,\sigma} \end{align*}
The case $r=4$ is used in \cite[Corollary~6]{r20}. Thus $$\|P_t-E_\sigma\|_{2\to r}\leq \|\sigma^{-1}\|_\infty^{1/2-1/r}e^{-\lambda t}.$$
Lemma~\ref{lem:2-8} and Proposition~\ref{cor:2-4} give HC $2\to r$ and EC with
coefficient $2/r$ whenever
\begin{equation}\label{eq:3-23}
 t\geq\frac{\log(r-1)+(1-2/r)\log\|\sigma^{-1}\|_\infty}{2\lambda}.
\end{equation}
For $r(t)=1+e^{2\alpha t}$, this condition becomes
\begin{equation}\label{eq:3-26}
 (\lambda-\alpha)t\geq\frac{\log\|\sigma^{-1}\|_\infty}{2}\tanh(\alpha t).
\end{equation}
The elementary bound $\tanh(\alpha t)\leq\alpha t$ shows that \eqref{eq:3-26} holds for every $t\geq0$ whenever
\[
 0\leq\alpha\leq A:=\frac{2\lambda}{2+\log\|\sigma^{-1}\|_\infty}.
\]
In particular, choosing $\alpha=A$ gives
\[
 \|P_t\|_{2\to1+e^{2At},\sigma}\leq1\qquad(t\geq0).
\]
By the definition of $\alpha_H$ in \eqref{eq:3-2}, this implies $\alpha_H\geq A$.
Together with $\alpha_1\geq\alpha_H/2$ from Lemma~\ref{lem:3-local}, this proves \eqref{eq:3-24}.
\end{proof}
The above theorem gives HC and entropy decay using only
coercivity and the smallest eigenvalue of $\sigma$, without KMS symmetry or $L_p$-regularity. 
A similar estimate for the LSI constant $\alpha_2$ was obtained in \cite{r20} for KMS-symmetric semigroups through Rothaus's lemma. Using this estimate to bound the MLSI constant $\alpha_1$ would require $L_p$-regularity.
In particular, our bound \eqref{eq:3-24} improves the 
qualitative MLSI conclusion
for primitive KMS-symmetric semigroups that also appears in
\cite[Theorem~3.1]{LWW2026} and \cite[Corollary~6.1]{GG2026}. Moreover, the logarithmic loss $\log\|\sigma^{-1}\|_\infty$ is asymptotically tight, as in the classical case.

A faster curve is available after a delay: for $0<\alpha<\lambda$, use
$\tanh(\alpha t)\leq1$ in \eqref{eq:3-26} to similarly obtain the HC and EC estimates whenever
\begin{equation}\label{eq:3-27}
 t\geq\frac{\log\|\sigma^{-1}\|_\infty}{2(\lambda-\alpha)}.
\end{equation}
For the tracial state $\sigma=I/d$ on $M_d$, a stronger reversible comparison is
recalled in \cite[Theorem~2.2]{MFW2016}:
$$\alpha_1\geq\alpha_2\geq \frac{2(1-2/d)\lambda}{\log(d-1)},$$ with coefficient $1$ at $d=2$.
Thus \eqref{eq:3-24} is a general faithful-state bound, not an optimal tracial bound.

The norm comparison at time zero $t=0$ can be replaced by a single HB estimate.
The interpolation argument below follows \cite[Theorem~5]{r20} and the classical
argument of \cite[Theorem~3.9]{r8}. Combining it with the centered HC criterion gives
the following bounds.  

\begin{theorem}[MLSI from a hyperbounded estimate]\label{thm:3-endpoint}
Let $P_t$ be a primitive QMS that is KMS-symmetric with respect to $\sigma$, and let $\lambda>0$ be its spectral gap. Assume that for some
$2<q\le \infty$, $t_q>0$, and $M_q\geq1$,
\begin{equation}\label{eq:3-28}
 \|P_{t_q}\|_{2\to q}\leq M_q.
\end{equation}
Then
\begin{equation}\label{eq:3-35}
 \alpha_1\geq\beta:=\frac{(1-2/q)\lambda}{2[\lambda t_q+\log M_q+1-2/q]}.
\end{equation}
Moreover, the LSI and HC constants satisfy
\begin{equation}\label{eq:3-36}
 \alpha_2\geq2\beta,\qquad \alpha_H\geq\beta.
\end{equation}
\end{theorem}

\begin{proof}
Set $b_q=1-2/q$. We use the Stein--Weiss argument of
\cite[equations~(18)--(19)]{r20}. For $K_z=e^{-zt_q\cL}$ on the strip
$0\leq\operatorname{Re}z\leq1$, KMS self-adjointness gives
$\|K_{iy}\|_{2\to2}=1$ and $\|K_{1+iy}\|_{2\to q}\leq M_q$.
Interpolation in the weighted $L_p$ scale therefore yields
\begin{equation}\label{eq:3-29}
 \|P_u\|_{2\to p(u)}\leq M_q^{u/t_q},\qquad
 p(u)=\frac2{1-b_qu/t_q},\qquad 0\leq u\leq t_q.
\end{equation}
For $2\leq r\leq q$ with $r<\infty$, put
\begin{equation}\label{eq:3-32}
 \theta_r=\frac{1-2/r}{b_q},\qquad u_r=t_q\theta_r.
\end{equation}
Since $\|P_T|_{\ker E_\sigma }\|_{2\to r}\leq M_q^{\theta_r}e^{-\lambda(T-u_r)}$ for $T\geq u_r$,
Lemma~\ref{lem:2-8} gives HC as soon as
\begin{equation}\label{eq:3-33}
 T\geq F(r):=t_q\theta_r+
 \frac{\theta_r\log M_q+\frac12\log(r-1)}\lambda.
\end{equation}
Together with Proposition~\ref{cor:2-4}, this proves the finite-exponent estimates
\begin{equation}\label{eq:3-34}
 \|P_T\|_{2\to r}\leq1,\qquad
 D(P_{T*}\rho\Vert\sigma)\leq\frac2rD(\rho\Vert\sigma).
\end{equation}
The function $F$ is continuously differentiable and strictly increasing on $[2,q)$,
with $F(2)=0$ and
\[
 F'(2)=\frac{\lambda t_q+\log M_q+b_q}{2b_q\lambda}.
\]
Its inverse $r(T)=F^{-1}(T)$ is a local HC curve, defined for $T\geq0$ sufficiently small,
whose initial slope is $r'(0)=4\beta$. Lemma~\ref{lem:3-local} proves all three bounds.
\end{proof}

The MLSI and LS$_2$ estimates use the initial slope directly; the extension to a full
HC curve gives the rate $\beta$. In the normalization of \cite{r20},
Theorems~\ref{thm:3-coercivity} and~\ref{thm:3-endpoint} give
$\alpha_2\geq2\lambda/(2+\log\|\sigma^{-1}\|_\infty)$ and $\alpha_2\geq2\beta$, respectively.
These are twice the corresponding bounds in \cite[Eqs.~(30), (11)]{r20}, after
reversing the generator sign. The interpolation method and the existence
of an all-time HC curve without $L_p$-regularity are already present in \cite{r20}.

\subsection{Delayed hypercontractivity and the lack of MLSI}\label{subsec:3-3}

Theorems~\ref{thm:3-coercivity} and~\ref{thm:3-endpoint} use $L_2$-decay of centered elements with
prefactor one. In finite dimensions, primitivity instead only guarantees an estimate
\begin{equation}\label{eq:3-14}
	\|P_t-E_\sigma \|_{2\to2}\leq Ce^{-\nu t},\qquad t\geq0,
\end{equation}
for some $C\geq1$ and $\nu>0$. Indeed, every positive rate below the smallest real part of
the spectrum of $\cL|_K$ is admissible after absorbing the Jordan polynomial factors
into $C$. This does not ensure $C=1$, and hence does not give the coercivity required
in Theorem~\ref{thm:3-coercivity}.

The norm comparison \eqref{eq:3-22} gives
\[
\|P_t-E_\sigma \|_{2\to r}
\leq C\|\sigma^{-1}\|_\infty^{1/2-1/r}e^{-\nu t},\qquad 2\leq r<\infty.
\]
For $r(t)=1+e^{2\alpha t}$, Lemma~\ref{lem:2-8} gives HC whenever
\[
(\nu-\alpha)t\geq\log C+
\frac12\log\|\sigma^{-1}\|_\infty\tanh(\alpha t).
\]
Since $\tanh(\alpha t)\leq1$, for every $0<\alpha<\nu$ and
\begin{equation}\label{eq:3-delay-static}
	t\geq\tau(\alpha):=
	\frac{\log C+\frac12\log\|\sigma^{-1}\|_\infty}{\nu-\alpha},
\end{equation}
we have
\begin{equation}\label{eq:3-delayed-hc-ec}
		\|P_t\|_{2\to1+e^{2\alpha t}}\leq1,\quad
		D(P_{t*}\rho\Vert\sigma)\leq\frac2{1+e^{2\alpha t}}D(\rho\Vert\sigma)
		\leq e^{-\alpha t}D(\rho\Vert\sigma).
\end{equation}
Here the entropy estimate follows from Proposition~\ref{cor:2-4} and
$2/(1+e^{2\alpha t})=e^{-\alpha t}/\cosh(\alpha t)$.
Nevertheless, such a delayed estimate does not imply a positive MLSI constant $\alpha_1>0$: it does not provide entropy decay from time zero.

The following qubit model makes this limitation explicit. It appears in
\cite[section ``Hypocoercivity and Physical Intuition'']{r9}; the same generator and
zero-entropy-production states occur in \cite[Propositions~4.1--4.2]{LWW2026}.

\begin{example}[Delayed decay with zero MLSI constant]\label{ex:3-zero}
	Let $\sigma=I/2$ on $M_2$ and $h,\gamma>0$. The QMS with generator
	\[
	-\cL_*(\rho)=-i[h\sigma_x,\rho]+\gamma(\sigma_z\rho\sigma_z-\rho)
	\]
	is primitive and satisfies \eqref{eq:3-14} for suitable $C,\nu$, and therefore obeys the delayed
	HC and EC estimates \eqref{eq:3-delayed-hc-ec}, but $\alpha_1=0$ and $\lambda_{\mathrm{KMS}}=0$.
    
    To verify these assertions, write $\rho=(I+x\sigma_x+y\sigma_y+z\sigma_z)/2$.
The Bloch equations are
\[
\dot x=-2\gamma x,\qquad
\frac{d}{dt}\begin{pmatrix}y\\z\end{pmatrix}
=\begin{pmatrix}-2\gamma&-2h\\2h&0\end{pmatrix}
\begin{pmatrix}y\\z\end{pmatrix}.
\]
The two eigenvalues of the displayed matrix are
$-\gamma\pm\sqrt{\gamma^2-4h^2}$ and have strictly negative real parts, so every
state converges to $I/2$. For every
$0<\nu<\gamma-\operatorname{Re}\sqrt{\gamma^2-4h^2}$, matrix exponential estimates
give a finite $C\geq1$ in \eqref{eq:3-14}, including the possible Jordan case.

For $\rho_a=(I+a\sigma_z)/2$, $0<|a|<1$, the dissipative part vanishes and the
Hamiltonian part has zero entropy production. Hence
\[
\mathcal I_{\cL}(\rho_a)=0,\qquad D(\rho_a\Vert I/2)>0,
\]
which gives $\alpha_1=0$. Finally,
$\operatorname{Re}\langle\sigma_z,\cL\sigma_z\rangle_\sigma=0$ and $L_2$
contractivity imply $\lambda_{\mathrm{KMS}}=0$.
\end{example}

\section{Entropy contraction and approximate tensorization}\label{sec:4}
For two conditional expectations, approximate tensorization easily implies entropy contraction
of their average. Theorem~\ref{thm:4-average} proves that the converse also holds up to an absolute constant.
The proof controls the Jensen--Shannon term in the entropy loss.
%Section~\ref{sec:4.5} relates the obstruction for products of conditional expectations to
%\cite{CCGP2025} and gives a simple partition example with explicit constants.

%\subsection{Averaged Channels}\label{sec:4.1}
Let $E_i:B(\mathcal{H})\to\cN_i$, $i=1,2$, be
conditional expectations preserving a common faithful state $\sigma$.
Namely, each $E_i$ is a unital completely positive map satisfying
\[ E_i(AXB)=AE_i(X)B \ \ \forall \ \ A,B\in\cN_i\  \text{ and } \qquad  E_{i*}\sigma=\sigma \ .\]
Each expectation is uniquely determined by the state $\sigma$ it preserves and its range $\mathcal{N}_i$, and acts as an orthogonal projection in $L_2(\sigma)$ \cite[Definition~1 and Proposition~1]{r4}.
The intersection range $\cN=\cN_1\cap\cN_2$ also admits a $\sigma$-preserving expectation
$E_{\cN}$. In this section, we focus on the average of the two conditional expectations
\[  \Phi_{\mathrm{av}}=\frac{E_1+E_2}{2} \]
as a positive self-adjoint contraction with fixed space $\cN$. Its powers $\Phi_{\mathrm{av}}^n$ converge to $E_{\cN}$.

For a state $\rho$, write
\begin{equation}\label{eq:4-1}
 \rho_i=E_{i*}\rho,\qquad \rho_{\cN}=E_{\cN*}\rho,\qquad
 D_{\cN}(\rho)=D(\rho\Vert\rho_{\cN}).
\end{equation}
The inclusion $\cN\subset \cN_i$ and the chain rule imply that
\begin{equation}\label{eq:4-2}
 E_{\cN}E_i=E_{\cN}=E_iE_{\cN},\qquad
 D(\rho||\rho_{\cN})=D(\rho\Vert\rho_i)+D(\rho_i||\rho_{\cN}),\qquad i=1,2,
\end{equation}
for all states $\rho$, including singular states
\cite[Lemma~3.4 of the arXiv version]{r12}.
It is known that the exact entropy tensorization inequality (also called entropy factorization)
\[ D(\rho\Vert\rho_{\cN})\le D(\rho \Vert\rho_1)+D(\rho \Vert\rho_2) \]
holds for every state $\rho$ if and only if the conditional expectations satisfy the commuting square condition
\[ E_1E_2=E_2E_1=E_{\mathcal{N}}.\]

We use the defect-free approximate tensorization inequality of
\cite[Definition~2, with defect $d=0$]{r4}, but allow any constant $c>0$ rather than imposing $c\ge1$.
We say that $(E_1,E_2)$ satisfies approximate tensorization with constant $c>0$, denoted by $\mathrm{AT}(c)$, if
\begin{equation}\label{eq:4-3}
 D(\rho||\rho_\cN)\leq c\bigl(D(\rho\Vert\rho_1)+D(\rho\Vert\rho_2)\bigr)
 \qquad\text{for every state $\rho$}.
\end{equation}
The optimal coefficient $c_{\mathrm{AT}}$ is the infimum of all admissible constants $c>0$, with $c_{\mathrm{AT}}=+\infty$ if none exists.
Thus $c_{\mathrm{AT}}$ is the sharp coefficient in \eqref{eq:4-3} and may be smaller than one.
For example, if $E_1=E_2=E_{\cN}\ne\id$, then $c_{\mathrm{AT}}=1/2$.
If $\cN=B(\mathcal H)$, both expectations are the identity and $c_{\mathrm{AT}}=0$; the formulas below involving entropy ratios or $1/c_{\mathrm{AT}}$ are used only when $\cN\ne B(\mathcal H)$.

The averaged map is
\begin{equation}\label{eq:4-4}
 \Phi_{\mathrm{av}}=\frac{E_1+E_2}{2},\qquad
 \Phi_{\mathrm{av},*}\rho=\frac{\rho_1+\rho_2}{2}, %\qquad T=E_{2*}E_{1*}.
\end{equation}
%The product applies $E_{1*}$ followed by $E_{2*}$ and has trace dual $E_1E_2$.
%In classical two-block systems, these are random-scan and systematic-scan updates;
%one product is a full sweep \cite[Algorithms~1--2]{r16}.
%Both maps preserve the $\cN$-marginal.
Entropy contraction for the average means
\begin{equation}\label{eq:4-5}
 D(\Phi_{\mathrm{av},*}\rho||\rho_{\cN})\leq\eta D(\rho||\rho_{\cN})
 \qquad\text{for every state $\rho$},\qquad 0\leq\eta<1,
\end{equation}
and is denoted by $\mathrm{EC}_{\cN}(\eta)$.
The reference state on both sides is $\rho_{\cN}$, which depends on the input $\rho$.
%We use the same notation for $T$ when it replaces $\Phi_{\mathrm{av},*}$.
%The complete versions require the same coefficients after replacing $E_i$ and
%$E_{\cN}$ by $E_i\otimes\id_{M_k}$ and $E_{\cN}\otimes\id_{M_k}$ for every $k$;
%the optimal complete AT coefficient is denoted by $c_{\mathrm{AT}}^{\mathrm c}$.

The known implication from AT to average EC follows from \eqref{eq:4-2} and convexity
\cite[Eqs.~(11)--(12)]{r7}:
\begin{equation}\label{eq:4-7}
 D(\Phi_{\mathrm{av},*}\rho||\rho_{\cN})
 \leq\frac12\sum_{i=1}^2D(\rho_i\|\rho_{\cN})
 \leq\left(1-\frac1{2c_{AT}}\right)D(\rho\| \rho_{\cN}).
\end{equation}
The converse requires an upper bound on the entropy lost in the averaging step.

\begin{theorem}[AT from EC]\label{thm:4-average}Let $E_i:B(\mathcal{H})\to\cN_i$, $i=1,2$, be
conditional expectations preserving a common faithful state $\sigma$, and let $E_{\cN}$ be the $\sigma$-preserving conditional expectation onto the intersection $\cN=\cN_1\cap \cN_2$.
Assume that the averaged map satisfies the following entropy contraction inequality for some $0<\eta<1$ and all states $\rho$:
\begin{align} D(\Phi_{\mathrm{av},*}\rho||\rho_{\cN})\leq\eta D(\rho||\rho_{\cN}).\label{eq:EC2}\end{align}
Then the following AT inequality holds:
\begin{equation}\label{eq:4-16}
 D(\rho||\rho_{\cN})\leq\frac{3}{2(1-\eta)}
 \left(D(\rho\Vert\rho_1)+D(\rho\Vert\rho_2)\right)
 \qquad\text{for every state $\rho$}.
\end{equation}
In particular,
\begin{equation}\label{eq:4-17}
 c_{\mathrm{AT}}\leq\frac{3}{2(1-\eta)}.
\end{equation}
\end{theorem}

\begin{proof}
Set $\bar\rho=(\rho_1+\rho_2)/2$. For states $A,B$, recall the quantum Jensen--Shannon divergence \cite[Section~1.2]{r22}
\begin{equation}\label{eq:4-8}
 J(A,B)=\frac12D\left(A\,\middle\Vert\,\frac{A+B}{2}\right)
       +\frac12D\left(B\,\middle\Vert\,\frac{A+B}{2}\right).
\end{equation}
Donald's identity \cite[equation~(5)]{r19} reads
\begin{equation}\label{eq:4-9}
 \frac12D(A\Vert R)+\frac12D(B\Vert R)
 =D\left(\frac{A+B}{2}\,\middle\Vert\,R\right)+J(A,B),
\end{equation}
when $\supp A,\supp B\subseteq\supp R$.
Taking $A=\rho_1$, $B=\rho_2$, and $R=\rho_{\cN}$, and using \eqref{eq:4-2} gives
the entropy-loss identity
\begin{equation}\label{eq:4-12}
 D(\rho\|\rho_{\cN})-D(\bar\rho\|\rho_{\cN})
 =\frac12\bigl(D(\rho\Vert\rho_1)+D(\rho\Vert\rho_2)\bigr)+J(\rho_1,\rho_2).
\end{equation}

By Virosztek's metric theorem \cite[Theorem~1]{r22}, for all states $A,B,C$,
\begin{equation}\label{eq:4-10}
 \sqrt{J(A,C)}\leq\sqrt{J(A,B)}+\sqrt{J(B,C)}.
\end{equation}
Taking $R=B$ in \eqref{eq:4-9} and using the nonnegativity of relative entropy, we also have
\begin{equation}\label{eq:4-11}
 J(A,B)\leq\frac12D(A\Vert B),\qquad
 J(A,B)\leq\frac12D(B\Vert A).
\end{equation}
Using the original input $\rho$ as the intermediate state in \eqref{eq:4-10} yields
\[
 \begin{aligned}
 \sqrt{J(\rho_1,\rho_2)}
 \leq\sqrt{J(\rho_1,\rho)}+\sqrt{J(\rho,\rho_2)}
 \leq\sqrt{\frac{D(\rho\Vert\rho_1)}2}+\sqrt{\frac{D(\rho\Vert\rho_2)}2}.
 \end{aligned}
\]
Combining the above estimates with the entropy contraction assumption \eqref{eq:EC2}, we obtain
\begin{align}
 (1-\eta) D(\rho||\rho_{\cN})
 &\leq D(\rho||\rho_{\cN})-D(\bar\rho||\rho_{\cN})\label{eq:4-13}\\
 &\leq \frac{1}{2}\left(D(\rho\Vert\rho_1)+D(\rho\Vert\rho_2)\right) +J(\rho_1,\rho_2)\label{eq:4-141}\\
 &\leq \frac{1}{2}\left(D(\rho\Vert\rho_1)+D(\rho\Vert\rho_2)\right) +\frac{1}{2}\left(\sqrt{D(\rho\Vert\rho_1)}+\sqrt{D(\rho\Vert\rho_2)}\right)^2\label{eq:4-142}\\
 &= D(\rho\Vert\rho_1)+D(\rho\Vert\rho_2) +\sqrt{D(\rho\Vert\rho_1)D(\rho\Vert\rho_2)}\label{eq:4-143}\\
 &\leq\frac32\bigl(D(\rho\Vert\rho_1)+D(\rho\Vert\rho_2)\bigr),\label{eq:4-15}
\end{align}
where the first inequality follows from entropy contraction 
$D(\bar{\rho}||\rho_{\cN})\le \eta D(\rho||\rho_{\cN}) $. This completes the proof.
\iffalse
For the complete assertion, apply the same argument to
$E_i^{(k)}=E_i\otimes\id_{M_k}$ and $E_{\cN}^{(k)}=E_{\cN}\otimes\id_{M_k}$.
These expectations preserve $\sigma\otimes I_k/k$, and their intersection is
$\cN\otimes M_k$.
Writing $D_{\cN}^{(k)}(\omega)=D(\omega\Vert E_{\cN*}^{(k)}\omega)$, the hypothesis is
\begin{equation}\label{eq:4-19}
 D_{\cN}^{(k)}(\Phi_{\mathrm{av},*}^{(k)}\omega)
 \leq\eta D_{\cN}^{(k)}(\omega)
 \qquad\text{for every $k$ and every state $\omega$}.
\end{equation}
Every step above has the same constant at every matrix level, giving
\begin{equation}\label{eq:4-20}
 D_{\cN}^{(k)}(\omega)\leq\frac{3}{2(1-\eta)}
 \bigl(D(\omega\Vert E_{1*}^{(k)}\omega)+D(\omega\Vert E_{2*}^{(k)}\omega)\bigr).
\end{equation}
Thus $c_{\mathrm{AT}}^{\mathrm c}\leq3/[2(1-\eta)]$.\fi
\end{proof}

The above bound holds for an arbitrary intersection $\cN_1\cap\cN_2=\cN$. It also holds completely: a complete EC bound with one coefficient $\eta<1$
implies complete AT with the same constant $3/[2(1-\eta)]$. Approximate tensorization for noncommuting expectations
and its complete versions are studied in \cite{r4,r10}.
Gao and Rouz\'e \cite[Corollary~5.4 and Remark~5.5]{r10} also obtain index-independent
bounds from mixing in completely positive order; converting $L_2$ estimates into that
mixing condition introduces a conditional-expectation index.
Theorem~\ref{thm:4-average} instead starts with an entropy contraction coefficient.
Its conversion to AT has no additional dimension or index factor.

For $\cN\ne B(\mathcal H)$, the optimal entropy contraction coefficient and the optimal AT coefficient are given by
\begin{equation}\label{eq:4-6}
\begin{aligned}
 \eta_{\mathrm{av}}&=\sup_{\rho\neq \rho_{\cN }}
 \frac{D(\Phi_{\mathrm{av},*}\rho||\rho_{\cN})}{D(\rho||\rho_{\cN})}, \\
 c_{AT}&=\sup_{\rho\neq \rho_{\cN }}
 \frac{D(\rho||\rho_{\cN})}{D(\rho\Vert\rho_1)+D(\rho\Vert\rho_2)}.
\end{aligned}
\end{equation}
The denominator defining $c_{AT}$ is positive for $\rho\ne\rho_{\cN}$: if both defects vanish, both state expectations fix $\rho$, and convergence of the averaged powers gives $\rho=\rho_{\cN}$.
Moreover, \eqref{eq:4-2} shows that each defect is at most $D_{\cN}(\rho)$, so $c_{AT}\ge1/2$.
The above discussion shows that
\[  \frac1{2c_{AT}}\leq 1-\eta_{\mathrm{av}}\leq\frac3{2c_{AT}}. \]
Here we use the convention $1/(+\infty)=0$.
We can now derive AT from HC for the averaged map.

\iffalse
The denominator defining $c_{\mathrm{AT}}$ is positive on this set: if both defects vanish, both
state expectations fix $\rho$, and convergence of the averaged powers gives
$\rho=\rho_{\cN}$. Data processing gives $0\leq\eta_{\mathrm{av}}\leq1$.
Set $\delta_{\mathrm{av}}=1-\eta_{\mathrm{av}}$.
Dividing \eqref{eq:4-13}--\eqref{eq:4-15} by $D_{\cN}(\rho)>0$ and taking infima gives
\[
 \frac1{2c_{\mathrm{AT}}}\leq\delta_{\mathrm{av}}\leq\frac3{2c_{\mathrm{AT}}},
 \qquad 1/(+\infty)=0.
\]
Thus AT has a finite constant if and only if $\delta_{\mathrm{av}}>0$. In that case,
\begin{equation}\label{eq:4-18}
 \frac1{2\delta_{\mathrm{av}}}\leq c_{\mathrm{AT}}\leq\frac3{2\delta_{\mathrm{av}}}.
\end{equation}
If $\cN=\cM$, both expectations are the identity, all entropy defects vanish, and
$c_{\mathrm{AT}}=0$; the ratio definitions and reciprocal comparisons are not used.
The known implication \eqref{eq:4-7} also holds completely, by the same convexity argument.
\fi 
\begin{corollary}[HC implies AT] Let $E_i:B(\mathcal{H})\to\cN_i$, $i=1,2$, be
two $\sigma$-preserving conditional expectations with trivial intersection $ \cN_1\cap \cN_2=\mathbb{C}1$. Suppose that, for some $1<p<q<\infty$,
 $$\|\Phi_{\mathrm{av}}\|_{p\to q,\sigma}\leq1. $$
Then we have 
\[D(\Phi_{\mathrm{av},*}\rho||\sigma )\le \min\{\frac{p}{q},\frac{q'}{p'}\} D(\rho||\sigma ),\] 
where $p'=p/(p-1)$ and $q'=q/(q-1)$. This further implies
\begin{equation}\label{eq:4-22}
 c_{\mathrm{AT}}\leq\frac{3}{2\bigl(1-\min\{\frac{p}{q},\frac{q'}{p'}\}\bigr)}.
\end{equation}
\end{corollary}
Here we use the fact that the average $\Phi_{\mathrm{av}}$ is KMS self-adjoint. For each fixed system with scalar intersection, this implies
$\|\Phi_{\mathrm{av}}-E_\sigma\|_{2\to2}<1$.
Together with any finite $2\to q$ endpoint, Theorem~\ref{prop:2-9} also supplies an HC
exponent and hence a finite AT constant.

\noindent {\bf Acknowledgments. }
The authors thank Cambyse Rouz\'e for helpful discussions on the idea of deriving entropy contraction from hypercontractivity.
Li Gao and Lijun Wang are supported
in part by the National Natural Science Foundation of China (grant no. 12401163) and the
Department of Science and Technology of Hubei Province (project nos. 2025EHA041 and
2025AFA044). The authors acknowledge assistance from ChatGPT in developing the proof strategy for deriving approximate tensorization from entropy contraction of the averaged channel in Section 4. All arguments were independently verified by the authors.

\bigskip
\begingroup
\small
\setlength{\parindent}{2em}
\setlength{\parskip}{0pt}

\indent\textsc{Li Gao}\par
\indent\textit{School of Mathematics and Statistics,
Wuhan University, Wuhan 430072, China}\par
\indent\textit{Wuhan Institute of Quantum Technology,
Wuhan 430075, China}\par
\indent\textit{E-mail address: }
\texttt{gao.li@whu.edu.cn}\par

\medskip

\indent\textsc{Lijun Wang}\par
\indent\textit{School of Mathematics and Statistics,
Wuhan University, Wuhan 430072, China}\par
\indent\textit{E-mail address: }
\texttt{lijun\_wang@whu.edu.cn}\par

\endgroup

\end{document}